\documentclass[12pt]{elsarticle}
\usepackage{preamble}
\begin{document}
\begin{frontmatter}

\title{Axioms for parametric continuity of utility when the topology is coarse\tnoteref{mytitlenote}}
\tnotetext[mytitlenote]{ I thank the anonymous referees, Itzhak Gilboa, Sander Heinsalu, Atsushi Kajii, Ali Khan, Jeff Kline, Andrew McLennan, Aliandra Nasif, Wojciech Olszewski, John Quiggin, Maxwell Stinchcombe and Rabee Tourky for taking the time to read and provide feedback on earlier versions of this paper. 
I further thank Maxwell Stinchcombe for numerous conversations that helped to enrich my understanding of the subject matter of this paper. 
}

\author{Patrick H. O'Callaghan\fnref{myfootnote}}
\address{School of Economics, University of Queensland}
\fntext[myfootnote]{Email: \texttt{p.ocallaghan@uq.edu.au}, Orcid Id: 0000-0003-0606-5465.}

\begin{abstract}
In economics we often take as primitive a collection of preference orderings (on actions or alternatives) indexed by a parameter. 
Moreover, it is often useful to represent such preferences with a collection of utility functions that is continuous in the parameter. 
Existing representation theorems assume that the topology on the  parameter space is metrizable.
This excludes settings where the topology is coarse \eg~the weak$^*$ topology on a set of probability measures or the product topology on many function spaces. 
Yet such spaces are often normal (disjoint closed sets can be separated). 
We introduce an axiom on preferences for parametric continuity when actions are countable and the parameter space is normal. 
Utility is jointly continuous on actions $\times$ parameters when actions have the discrete topology. 
\end{abstract}


\end{frontmatter}

\section{Introduction}\label{sec_Introduction}
A (topological) space is \emph{metrizable} if there exists a metric such that 1.~every open ball of the metric is also open in the space and 2.~every open set in the space can be written as a  union of open balls of the metric.
For instance, a discrete space (where every subset is open) is metrizable via the metric that assigns distance one to pairs of distinct points  and distance zero otherwise. 

In the present article, we are interested in utility representations of preference orderings that depend on a parameter. 
In particular, we identify the class of preferences that have a utility representation that is continuous in the parameter even when the parameter space is nonmetrizable. 
This feature distinguishes the present results from similar results (\citet{Hildenbrand_Joint_continuity,Mas-Colell_Joint_continuity,Levin_Joint_continuity}) where metrizability is a key assumption. 

Our results are intended for settings where the topology is too coarse to support a continuous metric (a failure of property 1 of metrizable spaces). 
This distinguishes the present results from the recent work of \citet[Theorem 4.1]{CCM_Joint_continuity} where the parameter space is allowed to be \emph{submetrizable}. 
A submetrizable space accommodates failures of property 2  by assuming the  topology can be coarsened (by excluding open sets)  in such a way that a metrizable space obtains. 

\subsection{Motivation}
In economics, the standard cause for a parameter space to be   nonmetrizable is a failure of property 1. 
The most common example is that of the product topology on a Cartesian product of uncountably many factors.  
Such parameter spaces arise: in the setting of Bayesian games  where the universal type space is formalised (\citet{MZ_Bayesian_analysis_for_games,DFM_Topologies_on_types}; in  the general equilibrium with infinite-dimensional commodity spaces (\citet{MZ_Infinite_dimensional_commodity_spaces}); and  in settings with continuous time stochastic processes.  

A simple  example of a nonmetrizable space is the following. 
\begin{example}\label{eg_Uncountable_product}
For some uncountable set $S$, let $\Theta=\{0,1\}^S$ be endowed with the product topology. 
Then the basic open sets in $\Theta$ are \emph{cylinder sets} $G=\bigtimes_{s\in S} G_s$ with  all but finitely many factors $G_s=\{0,1\}$. 
Let $d$ be a metric on $\Theta$. 
Then for any $\zeta\in \Theta$, the singleton set $\{\zeta\}$ is the intersection over the collection of balls $\left \{\theta: d(\theta,\zeta)<1/n\right\}$ such that $n=1,2,\dots$. 
This implies that these balls are finer than cylinder sets, for the intersection of any countable collection of cylinder sets has  uncountably many factors equal to $\{0,1\}$. 
Thus, $d$ is necessarily a discontinuous metric on $\Theta$ (with the product topology). 
\end{example}
The issue raised by \cref{eg_Uncountable_product} is particularly relevant when  preferences over actions are indexed by parameters. 
That is when preferences are described by a family of  binary relations (over actions) that vary over the parameter space. 
An  extension of \cref{eg_Uncountable_product} exposes the problem. 
\setcounter{example}{0}
\begin{example}[continued]\label{eg_Preferences}
Suppose there are just two actions $a$ and $b$. 
Moreover, suppose that $b$ is strictly preferred to $a$ on the open set of parameters $\theta$ such that $\theta\neq \zeta$ and that  $a$ is indifferent to $b$ at $\zeta$. 
Now such preferences are seemingly innocuous since $b$ is weakly preferred to $a$ at each $\theta\in\Theta$. 
Indeed it is trivial to find a family of utility functions that represents such preferences: let the utility of $a$ be normalised to zero for every $\theta\in \Theta$ and let the utility of $b$ equal to one for $\theta\neq\zeta$ and zero otherwise. 
Clearly, as a function on $\Theta$, the utility of $b$ is discontinuous. 
Indeed, similar to the way every metric on $\Theta$ is discontinuous,  every utility representation of such preferences is discontinuous in the parameter. 
\end{example}
Many theoretical results rely on parametric continuity of the utility representation. 
For instance, this is a necessary requirement for the envelope theorems of \citet{MS_Envelope_theorems}, where, similar to the present article, the main results do not require that actions form a topological space. 

This provides the primary motivation for the present article. 
Although there is no hope for parametric continuity of utility for the preferences of \cref{eg_Uncountable_product}, this is no justification for avoiding nonmetrizable spaces altogether. 
Indeed there is a host of perfectly reasonable preferences, indexed by the parameter space of \cref{eg_Uncountable_product},  that  have a parametrically continuous utility representation. 
And yet the preferences  of \cref{eg_Uncountable_product} are also reasonable. 

\subsection{Overview}
Our goal is to identify conditions, both on the parameter space and on preferences, that   distinguish those preferences that possess a parametrically continuous representation and those which do not. 
The basic results in this respect are  \cref{thm_Main} and \cref{prop_Upper_bound}. 
These results reveal how far, beyond metrizability of the parameter space, we may go without strengthening a standard axiom for continuity. 
\emph{Pairwise stability} requires that the set of parameters such that one action is strictly preferred to another is open. 

Pairwise stability is standard in the sense that, when actions form a discrete space, it is equivalent to the more commonly encountered \emph{closed graph} axiom we define on \cpageref{CG}. 
 The closed graph axiom has a long history in the literature that derives a utility representation that is jointly continuous on the product of  actions and parameters (\citet{Hildenbrand_Joint_continuity,Mas-Colell_Joint_continuity,Levin_Joint_continuity,CCM_Joint_continuity}). 
 
\Cref{thm_Main} holds for parameter spaces that are perfectly normal. 
Whilst the space of \cref{eg_Uncountable_product} is normal (disjoint closed sets can be separated) it is not perfect (because it contains closed sets that cannot be written as a countable intersection of open sets). 
\Cref{prop_Upper_bound} shows that every parameter space that is not perfectly normal supports preferences that, like those of \cref{eg_Uncountable_product}, have no continuous representation. 
We therefore strengthen pairwise stability to \emph{perfect pairwise stability}: the set of parameters such that one action is strictly preferred to another is not only open, it is also the union of a countable collection of closed sets. 
Perfect pairwise stability rules out preferences such as those of \cref{eg_Uncountable_product}. 
One motivation for this exclusion  is that such preferences can never be identified by asking a countable number of questions about preferences on the closure of basic sets (cylinder sets in that case) such that strict preference for action $b$ over action $a$ holds uniformly.

Our main theorem then holds for a parameter space that is normal. 
Together with a standard ordering axiom, perfect pairwise stability is shown to be necessary and sufficient for a utility representation that is continuous in the parameter. 
The following three steps outline the proof that our axioms are sufficient for the desired representation. 

First, perfect pairwise stability ensures that each pair of distinct actions determines a continuous pseudometric on the parameter space. 
This pseudometric may assign distance zero to a pair of distinct parameters, but only if the two actions are found to be indifferent at both parameters. 

Second, because the set of actions is countable,  the above pseudometrics combine to form a new  pseudometric. 
This too is continuous on the parameter space. 
The latter assigns distance zero to a pair of parameters only if preferences coincide. 
(In \citet[p.287]{DFM_Topologies_on_types} construct a similar pseudometric on players' type space in Bayesian games.) 

Third, although our pseudometric will not typically be a metric, the fact that it is constructed using preferences allows us to induce a quotient space such that pairs of points are only identified if preferences coincide. 
On this quotient space, our pseudometric is a metric, so that \cref{thm_Main} applies. 
Finally, the fact that the pseudometric is continuous and generated by preferences allows us to extend the representation to the original parameter space.

\Cref{sec_Model} introduces the relevant topological notions  for the parameter space and  the axioms on  preferences and then concludes with minor results that relate the axioms. 
\Cref{sec_Parametric_continuity} presents our main theorems for parametrically continuous utility representations where  actions need not form a topological space. 
In \cref{sec_Joint_continuity}, we consider a discrete action space  and derive representations that are  jointly continuous on actions $\times$ parameters. 
With a remark on obstacles to representing preferences on uncountable action sets with normal parameter spaces and the appendix of proofs, we conclude.

\section{Parameters, preferences and  axioms}\label{sec_Model}

Recall that a set $X$ becomes a topological space provided it is endowed with collection $\tau$ of  subsets such that $G\in \tau$ is \emph{open} (in $X$). 
$\tau$ is closed under arbitrary unions and finite intersections. 
As usual, reference to  $\tau$ is typically suppressed.  
The coarsest or weakest possible topology is the trivial  topology, where the only open sets are $X$ and the empty set $\emptyset$. 
At the other extreme, $X$ is \emph{discrete} if its topology is such that every subset of $X$ is open. 
\subsection{The parameter space} \label{sec_Topology}

Let $\Theta$ denote a nonempty topological space of  \emph{parameters}. 
A basic topological requirement  that we take for granted (unless stated otherwise) is that every singleton $\left\{\theta\right\}$ is the closed complement of set that is open in $\Theta$. 
This is the $\mc T_1$ separation axiom in topology. 
In the present section, we will appeal to the stronger, Hausdorff, assumption: if $\theta\neq \theta'$, then there exist disjoint open neighbourhoods $N$ and $N'$ of $\theta$ and $\theta'$ respectively. 

In \cref{sec_Parametric_continuity}, we will require that $\Theta$ is also \emph{normal}: if $F$ and $F'$ are disjoint closed subsets of $\Theta$, then there exist disjoint open neighbourhoods $G$ and $G'$ of $F$ and $F'$ respectively. 
Clearly, every normal space that is $\mc T_1$ is also Hausdorff.  
Conversely, every compact Hausdorff space is  $\Theta$ normal. 

The final separation  axiom that we consider is the following: $\Theta$ is \emph{perfect} if every closed subset of $\Theta$ is the intersection of a countable collection of open sets. 
When a space is both perfect and normal it is perfectly normal.  

Every metrizable space is perfectly normal. 
Interesting mathematical examples that perfectly normal and nonmetrizable include the split interval and the long line (see \citet{Johnson_Compact_nonmetrizable} and \citet{SS_Counterexamples} respectively). 
Whilst such spaces rarely find their way into economics, economists frequently adopt spaces that are  nonmetrizable, yet normal. 
First, note that every ordered space, where the topology is  generated by the open intervals of a linear order, is normal. 
Second, every compact Hausdorff space is normal. 
Then, since every product of compact Hausdorff spaces is compact Hausdorff, many useful spaces that are nonmetrizable are normal. 
Third, the weak$^*$ topology on a set of probability measures is often normal and  nonmetrizable. 
The following example shows that this is the case for the parameter space that \citet{GS_Context_EU} adopt to index preferences. 

\begin{example}
Let $\Sigma$ denote an algebra of subsets of $S$ and let $B(\Sigma)$ be the set of bounded $\Sigma$-measurable real-valued functions on $S$ (with the uniform or $\sup$ norm). 
Via the Riesz representation theorem, each continuous linear functional on $B(\Sigma)$ is of the form $\int_S \cdot \diff \mu$ for some bounded and finitely additive measure $\mu: \Sigma\rightarrow \R$. 
The latter space of measures is commonly denoted as $\textup{ba}(\Sigma)$. 
The weak$^*$ topology on $\textup{ba}(\Sigma)$ is the coarsest topology such that the map  $\mu\mapsto \int_S f\diff \mu$ is continuous on $\textup{ba}(\Sigma)$ for every $f\in B(\Sigma)$. 
The set $\Delta$ of finitely additive probability measures on $\Sigma$  coincides with the set $\left\{\mu: \int_S 1\diff \mu=1\right\}$. 
Since $\{1\}$ is closed in $\R$ and $\mu\mapsto \int_S 1\diff \mu$ is continuous, $\Delta$ is weak$^*$-closed in $\textup{ba}(\Sigma)$. 
Then, since every closed subset of a compact set is compact,  $\Delta$ is weak$^*$-compact by the Banach Alaoglu theorem:   the closed unit ball of $\textup{ba}(\Sigma)$ is weak$^*$-compact. 
Finally, by the Hahn-Banach theorem, $\Delta$ is also Hausdorff. 
Thus, with the weak$^*$ topology, $\Delta$ is a normal space. 
\end{example}
Although normal spaces are a substantial generalisation of metrizable spaces, Stone's theorem (\citet[Theorem 4]{Stone_Product_spaces})  tells us that the product of uncountably many metrizable spaces fails to be normal unless all but countably many factors are compact. 
Thus, when $I$ is the unit interval, $\R^I$ fails to be normal. 
Whilst this shows that our restriction to normal spaces is substantial, it is possible to address this particular issue by considering the one-point  compactification of each factor. 
(See \cite[Ch.16]{Taylor_Measure_theory_and_integration}, where this step is taken in the derivation of the  measure of a standard Brownian motion.) 

Finally, one might argue that normal topologies are not actually that coarse. 
After all it is easy to imagine coarse spaces where the $\mc T_1$ separation axiom fails to hold. 
Perhaps in a  game where agents only have access to partitions of $\Theta$; or if $\Theta$ is a weakly ordered space (weak orders are defined in the following subsection). 
Fortunately, this paper provides methods for accommodating such spaces. 
In essence, it is possible to work with such spaces provided it is possible to pass to a quotient space where  the axioms on preferences we describe next hold. 
Indeed a similar step is key to the proof of  \cref{thm_Normal}.  




\subsection{Ordering of actions given the parameter}\label{sec_Preferences}
Let  $A$ denote a nonempty and countable set of \emph{actions} or \emph{alternatives}.  
In the absence of contrary statements, no external topology on $A$ is imposed. 
The primitive data regarding the decision maker's preferences  comes in the form of statements such as ``at $\theta$, action $b$ is strictly preferred to action $a$". 
\emph{Preferences at  $\theta$} are summarised by a subset $\prec_\theta$ of $A^2= A\times A$. 
\emph{Preferences} are defined to be a family $\left\{\prec_\Theta\right\}\defeq\left\{\prec_\theta\,:\theta\in \Theta\right\}$ of binary relations on $A$. 
Equivalently, the map $\theta\mapsto \,\prec_\theta$ is a correspondence on $\Theta$ with values in $A^2$. 

When  strict preference is primitive, the following condition is standard. 
\begin{taggedaxiom}{$\mathcal{O}$}\label{Order}
Preferences satisfy \emph{asymmetry} and \emph{negative transitivity} of $\prec_\theta$ for every  $\theta\in \Theta$. That is, respectively,
\begin{enumerate}[label=$\mathcal{O}_{\arabic*}$]
\item\label{Asy}
For every $a,b\in A$, the set $\left\{\theta: \text{$a\prec_\theta b$ and $b\prec_\theta a$} \right\}$ is empty. 

\item\label{NT}
For every $a,b,c\in A$ and $\theta\in \Theta$, if $a\prec_\theta b$ then   $a\prec_\theta c$ or $c\prec_\theta b$. 
\end{enumerate}
\end{taggedaxiom}
\ref{Order} ensures $\prec_\theta$ is an  \emph{ordering} of $A$ (alternative terms include \emph{asymmetric weak ordering} \citep{Fishburn_Utility_theory_Book,Fishburn_Foundations_of_EU} and \emph{strict weak ordering}). 
For $a$ and $b$ that $\prec_\theta$ finds incomparable  (neither $a\prec_\theta b$ nor $b\prec_\theta a$) we write $a\sim_\theta b$. 
 \ref{NT} ensures that $\left\{\sim_\Theta\right\}$ is a collection of  \emph{indifference} (transitive incomparability) relations  on $A$. 
Recall that $\sim_\theta$ is \emph{transitive} provided that $a\sim_\theta b$ and $b \sim_\theta c$ together  imply  $a\sim_\theta c$ for every $(a,b,c)\in A^3$. 
Finally, weak preference $\precsim_\theta$ is then the union of $\prec_\theta$ and $\sim_\theta$. 
By construction,  $\precsim_\theta$ is \emph{complete}: for every $(a,b)\in A^2$,  $a\precsim_\theta b$ holds or $b\precsim_\theta a$ holds. 
Together, \ref{Asy} and \ref{NT} ensure $\precsim_\theta$ is also transitive. 
\subsection{Perfect pairwise stability of strict preference}\label{sec_Pairwise_stability}
The following \emph{closed graph} axiom provides the traditional route to parametric continuity from preferences. 
With the understanding that $A$ is also a topological space and that $\Theta\times A^2$ is endowed with the product topology, 
the following formulation is due to \citet{Hildenbrand_Joint_continuity}. 
\begin{taggedaxiom}{$\mc {CG}$} \label{CG}
The set $\left\{(\theta,a,b): a\precsim_\theta b\right\}$ is closed in $\Theta\times A^2$. 
\end{taggedaxiom}

In contrast the  axioms that  we now present do not require that $A$ is a topological space.
We say that strict preference is \emph{pairwise stable at $\theta$} provided, for every $a,b\in A$ such that $a\prec_\theta b$, there exists an open neighbourhood $N$ of $\theta$ in $\Theta$ such that $a\prec_\eta b$ for every $\eta$ in $N$. 
Since any union of open sets is open, the following axiom captures \emph{pairwise stability} of strict preference.
\begin{taggedaxiom}{$\mathcal{PS}$}\label{PS}
For every $a,b\in A$, the set $\{\theta: a \prec_\theta b\}$ is open in $\Theta$.
\end{taggedaxiom}
To our knowledge, \ref{PS} first appears in the literature on decision theory in \citet{GS_Context_EU}. 
When $\Theta$ is not perfect, preferences may satisfy \ref{PS} and yet  fail to be what we call \emph{perfectly pairwise stable}.

\begin{taggedaxiom}{$\mc {PS^*}$}\label{Perfect}
For every $a,b\in A$,  $\{\theta: a\prec_\theta b\}$ is open in $\Theta$ and equal to the union of a countable collection of sets that are closed in $\Theta$. 
\end{taggedaxiom}

When preferences satisfy \ref{Perfect}, there exists a countable collection $\left\{F_n\right\}$ of closed subsets of $\left \{\theta: a\prec_\theta b\right\}$, such that, for each $\eta\in \left \{\theta: a\prec_\theta b\right\}$, there exists $m\in \mbb N$ such that $ \eta\in F_m$. 
For instance, if $\Theta$ were levels of  wealth, $F_m$ would be a finite union of closed intervals in $\R$. 
To our knowledge, \ref{Perfect} is novel and, as we will show,   it is the key axiom to parametric continuity when the parameter space is normal, but not perfectly so.  




\subsection{Relating the axioms}\label{sec_Discussion_axioms}
In the present subsection, in order to compare \ref{CG} with \ref{PS} and \ref{Perfect},  we assume that $A$ is a topological space. 
This allows us to speak of continuity properties of the  correspondence $\theta\mapsto \,\prec_\theta$. 
First, recall that it is \emph{lower hemicontinuous} if $\left\{\theta:\, \prec_\theta\cap \,B\neq  \emptyset\right\}$ is open for every open $B\subseteq A^2$.

\begin{remark}[proof on \cpageref{proof_lhc}]\label{prop_lhc}
If preferences satisfy \ref{PS}, then $\theta\mapsto \,\prec_\theta$  is lower hemicontinuous. 
When $A$ is discrete, the converse is also true. 
\end{remark}
\begin{proofatend}[\textsc{Proof of \cref{prop_lhc} of \cpageref{prop_lhc}}]\label{proof_lhc}

The map $\theta\mapsto \,\prec_\theta$ defines a correspondence on $\Theta$ with values in the power set of $A\times A$. This  is \lhc~provided that, for every open $G\subseteq A\times A$, the set $\{\theta: \,\prec_\theta \hskip-2pt\cap \,G \neq \emptyset\}$ is open. This latter set is just the union of $\{\theta:a\prec_\theta b\}$ such that $a\times b\in G$. Thus, \ref{PS} implies  \lhc~of $\theta\mapsto \,\prec_\theta$. 

When $A$ is discrete, every $B\subseteq A\times A$ is both open and closed. 
If $\theta\mapsto\,\prec_\theta$ is \lhc, take $B=\{a\times b\}$. 
Then $\{\theta: \prec_\theta \cap \,\,B \neq \emptyset\}$ is open and equal to $\{\theta:a\prec_\theta b\}$. 
Thus, when $A$ is discrete the converse holds, as required. 
\end{proofatend}

Even when \ref{Order} holds, \ref{PS} does not imply that the weak preference correspondence $\theta\mapsto\,\,\precsim_\theta$ is \emph{upper hemicontinuous}: $\left\{\theta:\,\, \precsim_\theta \cap \, F\neq \emptyset\right\}$ is closed for every closed $F\subseteq A^2$. 
Indeed, let $F$ be any closed, infinite subset of $A^2$, then  although \ref{Order} and \ref{PS} together ensure the set $\{\theta: a\precsim_\theta b\}$ is closed, the union over the pairs $a\times b\in F$, need not be closed. 
On the other hand, we have

\begin{remark}[proof on \cpageref{proof_Closed_graph}]\label{prop_Closed_graph}
Let $A$ be discrete, let $\Theta$ be Hausdorff and let preferences satisfy \ref{Asy}.
Then \ref{PS} holds if and only if \ref{CG} holds. 
\end{remark}
\begin{proofatend}[\textsc{Proof of \cref{prop_Closed_graph} of \cpageref{prop_Closed_graph}}]\label{proof_Closed_graph}
Suppose \ref{PS} holds.
Let $D$ be a directed set and let $\eta_\nu\times a_\nu\times b_\nu\rightarrow \eta\times a\times b$ be a net such that $a_\nu\precsim_{\eta_\nu} b_\nu$ for each $\nu\in D$. 
Then there exists $\mu\in D$ such that for every $\nu\geq \mu$, $\eta_\nu\times a_\nu\times b_\nu=\eta_\nu\times a\times b$. 
Now suppose that $b\prec_\eta a$, so that the graph of $\theta\mapsto \,\,\precsim_\theta$ is not closed. 
Then by \ref{PS}, there exists a neighbourhood $N$ of $\eta$ such that $a\prec_\theta b$ for every $\theta\in N$. 
Since the net converges to $\eta\times a\times b$, there exists $\mu'$ such that $\eta_\nu\in N$ for every $\nu\geq \mu'$. 
But then for every $\nu\geq \max\{\mu,\mu'\}$, we have both $a\precsim_{\eta_\nu} b$ (because $\eta_\nu\times a\times b$ belongs to the graph of $\precsim_\Theta$)  and $b\prec_{\eta_\nu}  a$ (because $\eta_\nu\in N$), a contradiction of \ref{Asy}.

Now suppose that the graph of $\theta\mapsto \,\,\precsim_\theta$ is closed. 
Then for fixed $a,b\in A$, the set $\left\{(\theta,a,b): a\precsim_\theta b\right\}$ is closed. 
By \ref{Asy}, this is equivalent to \ref{PS}.
\end{proofatend}

Experiments in the field and the laboratory frequently deal with discrete action sets. 
Since (perfect) pairwise stability is easily understood, 
\cref{prop_Closed_graph} and our final remark provide a way to evaluate the hypothesis that preferences satisfy \ref{CG}. 
\begin{remark}[proof on \cpageref{proof_Metrizable_Theta}]\label{prop_Metrizable_Theta}
Let $A$ be discrete, let $\Theta$ be perfectly normal and let preferences satisfy \ref{Asy}.
Then \ref{PS}, \ref{Perfect} and \ref{CG} are equivalent. 
\end{remark}
\begin{proofatend}[\textsc{Proof of \cref{prop_Metrizable_Theta} of \cpageref{prop_Metrizable_Theta}}]\label{proof_Metrizable_Theta}
When $\Theta$ is perfect, every closed set is the intersection of a countable collection of open sets. 
Equivalently, every open set is the countable union of closed sets. 
Fix $a,b\in A$ and let $G=\left\{\theta: a\prec_\theta b\right\}$. 
Since $\Theta$ is perfect, $G=\bigcup_1^\infty F_n$ where $F_1,F_2,\dots$ is a collection of closed sets. 
Thus, \ref{PS} implies \ref{Perfect} whenever $\Theta$ is perfect. 
Clearly, the converse holds regardless of the topology on $\Theta$. 
The proof follows by \cref{prop_Closed_graph} and the fact  that every normal $\mc T_1$ space is Hausdorff. 
\end{proofatend}



\section{Parametric continuity of utility}\label{sec_Parametric_continuity}
The present section is purely about parametric continuity of utility. 
Thus, no topology (discrete or otherwise) is required of  $A$. 
Our main goal is \cref{thm_Normal}, where  we obtain a parametrically continuous utility representation for the case where $ \Theta$ is normal and preferences satisfy \ref{Perfect}. 

Recall that  a \emph{utility representation} of $\prec_\theta$ is a function $U(\cdot,\theta):A\rightarrow \R$ such that, for every $a,b\in A$,  $a\prec_\theta b$ if and only if $ U(a,\theta)<U(b,\theta)$.  
A \emph{(utility) representation of preferences} is a function $U: A\times \Theta\rightarrow \R$ such that  $U(\cdot,\theta)$ is a utility representation of $\prec_\theta$ for every $\theta\in \Theta$.   
The representation $U$ is \emph{parametrically continuous} if $U(a,\cdot)$ is continuous on $\Theta$ for every $a\in A$. 
\subsection{Perfectly normal parameter spaces}

Our first theorem takes us beyond (sub)metrizable parameter spaces whilst maintaining the standard continuity axiom \ref{PS}. 
\begin{theorem}\label{thm_Main} Let  $A$ be countable and let $\Theta$ be perfectly normal. 
Preferences have a parametrically continuous  representation if and only if  both \ref{Order} and \ref{PS} hold. 

\end{theorem}
\begin{proof}[\textsc{Proof of \cref{thm_Main}}]
The proof that the axioms are necessary is provided in the proof of \cref{thm_Normal}. 
The proof that the axioms are sufficient for a parametrically continuous representation proceeds by induction on $A$. 
Since the initial case (\cref{step_Initial}) is useful for the discussion that follows we present it in the main text. 
In the appendix, the inductive case appeals to Michael's selection theorem (a characterisation of perfectly normal spaces).
\begin{step}\label{step_Initial}
Let $A=\{a,b\}$. By \ref{Asy} and \ref{PS}, $F=\left\{\theta: a\sim_\theta b\right\}$ is closed in $\Theta$. 
Since $\Theta$  is perfect, there exists $\left\{G_n:n\in \mbb N\right\}$ of open sets satisfying $\bigcap_1^\infty G_n=F$. For each $n$, note that $F$ and  $\Theta\bs G_n$ are disjoint and the latter is also closed. 
Since $\Theta$ is normal, the Urysohn lemma  guarantees the existence of a continuous, real-valued function on $\Theta$ such that $f_n(\theta)=0$ on $F$, $f_n(\theta)=1$ on $\Theta\bs G_n$, and $0\leq f_n(\theta)\leq 1$ otherwise. 

Let $f=\sum_1^\infty 2^{-n}f_n$ and note that $f:\Theta\rightarrow [0,1]$ is the continuous and  uniform limit of partial sums $\sum_1^m 2^{-n}f_n$ as $m\rightarrow\infty$.  
Moreover, since every $\theta\in \Theta\bs F$ belongs to some $\Theta\bs G_n$, $f(\theta)=0$ if and only if $a\sim_\theta b$. 
Let $U(a,\cdot)$ be the zero function on $\Theta$. 
We obtain a utility function for each $\prec_\theta$ by taking
\begin{equation*}
U(b,\theta)\defeq\left \{
\begin{array}{rl}
f(\theta)& \text{if $a\prec_\theta b$,} \\
-f(\theta)&  \text{otherwise.}
\end{array}\right.
\end{equation*}
To complete the proof of \cref{step_Initial}, it remains to confirm that  $U(b,\cdot)$ is continuous. 
This follows almost immediately from continuity of $f$. 
In particular, suppose that $\theta$ satisfies $a\prec_\theta b$ so that $U(b,\theta)=f(\theta)$. 
Then  \ref{PS} ensures the existence of an open neighbourhood $N_\theta$ such that $a\prec_\eta b$ for every $\eta\in N_\theta$. 
Then $U(b,\cdot)=f(\cdot)$ on $N_\theta$. 
This confirms that $U(b,\cdot) $ is continuous on $N_\theta$. 
By a similar argument, the fact that $-f$ is continuous allows us to conclude that $U(b,\cdot)$ is continuous at every $\theta$ such that $b\prec_\theta a$. 
Finally, for $\theta$ such that $a\sim_\theta b$, note that 
$f(\theta)=-f(\theta)$. 
For any directed set $D$, let  $\left\{\theta_\nu: \nu\in D\right\}$ be a net that converges to $\theta$. 
Then  $U(b,\theta_\nu)$ converges to $U(b,\theta)$ since $f$ and $-f$ are continuous and have the same limit. 
By \ref{Order}, this accounts for every $\theta\in \Theta$. 

The proof of \cref{thm_Main} continues on \cpageref{step_Inductive}. 
\qedhere
\end{step}\phantom\qedhere
\end{proof}\vskip-20pt
\begin{proofatend}[\textsc{Remaining steps in the proof of \cref{thm_Main} of \cpageref{thm_Main}}]\label{proof_Main}
\setcounter{step}{1}
\begin{step}[\textsc{Inductive step}]\label{step_Inductive}
  
 Let $\{1,2, 3 \dots\}$ be an arbitrary enumeration of $A$, and let $[j]\subseteq A$ denote the first $j$ elements of the enumeration. Fix $j\in A$. 
By the induction hypothesis, there exists a parametrically continuous (utility) representation $U^{j-1}: [j-1]\times \Theta\rightarrow [-1,1]$. 

For each $a\in [j-1]$ take $U^j(a,\cdot)\defeq U^{j-1}(a,\cdot)$. 
The inductive step is complete if we can find a function $U^j(j,\cdot)$ that satisfies the properties of $f$ in the following version of Michael's selection theorem \citep[Theorem 3.1''']{Michael_Selection}. 

\begin{theorem*}[\citet{GS_Restatement}]\label{thm_Michael}
$\Theta$ is perfectly normal if and only if, whenever $g, h : \Theta\rightarrow \R$ are respectively upper and lower semi-continuous functions and $g \leq h$, there is a continuous $f : \Theta\rightarrow \R$ such that $g \leq f \leq h$ and $g(\theta) < f (\theta) < h(\theta)$ whenever $g(\theta)<h(\theta)$\,.
\end{theorem*}

In our setting, $g$ and $h$ will be envelope functions. 
To ensure they are well-defined, we introduce two fictional actions $\underline a$ and $\overline a$. 
These satisfy the property: $ \underline a \precsim_\theta k\precsim_\theta\overline a$ for all $(k,\theta)\in [j]\times\Theta$. 
Accordingly, we define $[j-1]'=[j-1]\cup \{\underline a, \overline a\}$, and let $U^j(\underline a, \cdot)\equiv -1$ and $U^j(\overline a, \cdot)\equiv +1$. 
Both are clearly continuous functions on $\Theta$.
Moreover, for all $\theta\in \Theta$, the following functions are well-defined.
\begin{align*}
g(\theta)\defeq &\max\left\{U^j(k,\theta): \text{$k\precsim_\theta j$ and  $k\in [j-1]'$}\right\},\\
h(\theta)\defeq &\min\left\{U^j(k,\theta): \text{$j\precsim_\theta k$ and $k\in [j-1]'$}\right\}.
\end{align*}
In the three claims that follow, we prove that $g$ and $h$ satisfy the conditions for Michael's selection theorem. 
The inductive step is then complete since $U^j$ is a parametrically continuous representation with values in $[-1,1]$.
\begin{claim}\label{lem_Upp_geq_low}
For all $\theta\in \Theta$, $g(\theta)\leq h (\theta)$\,.
\end{claim}
\begin{proof}[\textsc{Proof of \cref{lem_Upp_geq_low}}]
Fix $\theta$. 
By construction, there exist $k,l\in [j-1]'$ satisfying $g(\theta)=U^j(k,\theta)$ and $
h(\theta)=U^j(l,\theta)$. 
By definition, $k\precsim_\theta j$ and $j\precsim_\theta l$. 
By \ref{NT}, $k\precsim_\theta l$ and the inductive hypothesis then  ensures that $g(\theta)\leq h(\theta)$.
\end{proof}
\begin{claim}\label{lem_Upp=low}
For all $\theta\in \Theta:$ $g(\theta)=h(\theta)$ iff $k\sim_\theta j$ for some $k\in [j-1]$.
\end{claim}
\begin{proof}[\textsc{Proof of \cref{lem_Upp=low}}]
If $g(\theta)=h(\theta)$, then, by construction, there is some $k\in [j-1]'\cap\{l: l\precsim_\theta j\}\cap\{l:j\precsim_\theta l\}$.
By \ref{Asy}, for every such $k$, $k\sim_\theta j$. Conversely, if  $k\sim_\theta j$, then  both $k\precsim_\theta j$ and $j\precsim_\theta k$.
\end{proof}
\begin{claim}\label{lem_Usc}
$g:\Theta\rightarrow \R$ is upper semicontinuous.
\end{claim}
A symmetric argument to the one that follows, but with inequalities and direction of weak preference reversed, shows that $h$ is lower semicontinuous.
\begin{proof}[\textsc{Proof of \cref{lem_Usc}}]
Recall (or see \cite[p.101]{Kelley_Topology}) that $g$  is upper semicontinuous provided the set 
$\{\theta:r \le g(\theta)\}$ is closed for each $r\in \R$. Note that by the construction of $g$,
\begin{align*}
\{\theta: r\le g(\theta)\}=\bigcup_{k\in[j-1]'}\left(\{\theta: r \le U^j(k,\theta)\}\cap\{\theta:k\precsim_\theta j\}\right).
\end{align*}
Recall that the  finite union of closed sets is closed.
Moreover, since $U^j(k,\cdot)$ is continuous, $\{\theta:r \le U^j(k,\theta)\}$ is closed (preimage of a closed set is closed); and $\{\theta:k\precsim_\theta j\}$ is closed by \ref{Asy} and \ref{PS}. 
\end{proof}
\end{step}
\begin{step}[\textsc{The countably infinite case}]\label{step_Infinite}
 The above argument holds for each $j$ in $\mathbb N$.\footnote{I thank Atsushi Kajii for bringing this subtle issue to my attention.} For countably infinite $A$, we choose $ U: A\times \Theta \rightarrow \R$ such that its graph satisfies $\graph U = \bigcup_{j\in \mbb N} \graph U^j$. Since  Michael's selection theorem is used at each $j$,  for this step we appeal to the axiom of dependent choice. Alternatively, following \cite[p.23]{Kreps_Notes}, let $U(j,\cdot)=U^j(j,\cdot)$ for each $j\in \mathbb N$, and again appeal to the axiom of (dependent) choice.
\end{step}
This completes the proof of \cref{thm_Main}.
\end{proofatend}

\Cref{step_Initial} in the proof of \cref{thm_Main} yields the ``only if'' part of the following alternative definition for a perfectly normal space. 
(For the converse, note that if $f:X\rightarrow \R$ is a continuous function, then  $G_n=\left\{x: \lvert f(x)\rvert <\nicefrac{1}{n}\right\}$ is an open neighbourhood of $F=f^{-1}(0)$ for each $n\in \mbb N$ and $\bigcap_1^\infty G_n=F$.) 
\begin{definition}\label{prop_Zero_set}
 $X$ is perfectly normal if and only if every closed subset $F$ of $X$ is a \emph{zero set}. That is, for some continuous $f: X\rightarrow \R$, $f^{-1}(0)=F$.
 \end{definition} 
By \cref{prop_Zero_set}, every space that is not perfectly normal, contains a closed subset $F$ that is not a zero set. 
If this is the case, then there exist preferences with the property $\left\{\theta: a\sim_\theta b\right\}=F$ for some $a,b\in A$. 
Since such preferences are feasible whenever $\Theta$ fails to be perfectly normal, we have 
\begin{proposition}[proof on \cpageref{proof_Upper_bound}]\label{prop_Upper_bound}
If $\Theta$ is not perfectly normal, then there are preferences with no parametrically continuous representation that satisfy both \ref{Order} and \ref{PS}. 
\end{proposition}       
\begin{proofatend}[\textsc{Proof of \cref{prop_Upper_bound} of \cpageref{prop_Upper_bound}}]\label{proof_Upper_bound}
Let $A=\{a,b\}$ and suppose that $F=\{\theta:a\sim_\theta b\}$ for some closed set $F$ that is not a zero set. Such an $F$ exists whenever $\Theta$ fails to be perfectly normal. 
Since preferences satisfy \ref{Asy} and there are only two actions, there exists a representation of preferences. 
Take $U:A\times \Theta\rightarrow \R$ to be any such representation and define $f:\Theta\rightarrow \R$ to be the map $\theta\mapsto U(a,\theta)-U(b,\theta)$. 
Since $U$ is a representation, $f(\theta)=0$ if and only if $\theta\in F$. Thus $f^{-1}(0)=F$ and, since $F$ is not a zero set,  $f=U(a,\cdot)-U(b,\cdot)$ is discontinuous. 
By the algebra of continuous functions, at least one of $U(a,\cdot)$ and $U(b,\cdot)$ is discontinuous. 
\end{proofatend}

\subsection{Extension to normal parameter spaces}\label{sec_Normal}
In the following theorem, we drop the requirement that $\Theta$ is perfect and instead assume that preferences are perfectly pairwise stable. 
\begin{theorem}\label{thm_Normal}
Let $A$ be countable and let $\Theta$ be normal. 
Preferences have a parametrically continuous representation if and only if  both \ref{Order} and \ref{Perfect} hold. 
\end{theorem} 
\begin{proof}[\textsc{Proof of \cref{thm_Normal}}]\label{proof_Normal} 
The proof that \ref{Order} is necessary when $A$ is countable and $\Theta$ is a singleton follows from classical results \citep{Cantor_1895}. 
The more general case follows by applying the same argument pointwise. 
The proof that \ref{Perfect} (and a fortiori \ref{PS}) is necessary for a parametrically continuous representation is as follows.  
Let $U$ be a parametrically continuous utility representation. 
Fix $a,b\in A$ and let $f:=U(a,\cdot)-U(b,\cdot)$. 
Note that $f$ is continuous by continuity of  $U(a,\cdot)$ and $U(b,\cdot)$. 
Then for each $n\in \mbb N$, $G_n:=\left\{\theta: f(\theta)\leq \nicefrac{1}{n}\right\}$
is a closed subset of 
$$G:=\left\{\theta: a\prec_\theta b\right\}=\left\{\theta: f(\theta)<0\right\}.$$ 
The fact that $f$ is continuous also ensures that $\{\theta: f(\theta)<0\}$ is open. 
Finally, the proof that \ref{Perfect} is necessary follows from  the fact that $\bigcup_1^\infty G_n=G$.

In the remainder of this proof, we show that \ref{Order} and \ref{Perfect} are sufficient for a parametrically continuous representation. 
In \cref{lem_Pseudo} of sub\cref{sec_Pseudometric}, we show that \ref{Perfect}, allows us to construct a pseudometric $p$ on $\Theta$. 
Like a metric, the pseudometric $p:\Theta^2\rightarrow \R $ is continuous, nonnegative, symmetric and satisfies the triangle inequality. In contrast with a metric, the pseudometric $p$ may satisfy $p(\zeta,\eta)=0$ for $\zeta\neq \eta$. 
Such pairs are incomparable under $p$ and the collection of such pairs forms a binary relation $\bowtie$ on $\Theta$.   

\Cref{lem_Pseudo} shows that $\bowtie$ is an equivalence relation. 
This ensures that $\bowtie$ partitions $\Theta$ and that we may pass to the quotient space $\Theta_p\defeq \Theta_{/\bowtie}$. (This space identifies points that are incomparable under $p$.) 
This identification ensures that, with the open sets generated by $p$, $\Theta_p$ is a $\mc T_1$ space. 
In fact, $\Theta_p$ is perfectly normal because every  pseudometrizable $\mc T_1$ space is metrizable. 

By \cref{lem_Pseudo},  $\zeta \bowtie \eta$ implies that  $\prec_\zeta$ equals $\prec_\eta $. 
\Cref{thm_Main} then ensures the existence of a parametrically continuous representation $U_p:A\times \Theta_p\rightarrow \R$.  
Finally, we extend $U_p$  from $\Theta_p$ to $\Theta$ as follows. 
For each $a\in A$, let $f_p:\Theta\rightarrow \Theta_p$ be the projection from $\theta\in \Theta$ to its equivalence class in $\Theta_p$. 
Then, for each $a\in A$, let $g_a:=U_p(a,\cdot)$ and take $U(a,\cdot):=g_a\ccirc f_p$. 
Then $U: A\times \Theta\rightarrow \R$ is constant on each equivalence class of $\bowtie$, and parametric continuity of $U$ follows from continuity of  $p$. 
(Each equivalence class of $\bowtie$  is closed in $\Theta$.) 
\end{proof}

\subsection{Pseudometrics for nonmetrizable spaces}\label{sec_Pseudometric}
Whilst this subsection is essential to the proof of \cref{thm_Normal}, 
it is also motivated by the needs of standard tools in the analysis of policy  and in particular the envelope theorems of \citet{MS_Envelope_theorems} which require that $\Theta$ is a metric space and impose no topological structure on $A$. 
When $\Theta$ is nonmetrizable, we may use preferences to generate a pseudometric and work on the quotient space that we described in the proof of \cref{thm_Normal}. 
\begin{lemma}[proof on \cpageref{proof_Pseudo}]\label{lem_Pseudo}
Let $A$ be countable and let $\Theta$ be normal. 
If \ref{Order} and \ref{Perfect} hold, then there exists a continuous pseudometric $p:\Theta^2\rightarrow \R_+$ such that, for every $\zeta, \eta\in \Theta$,  $p(\zeta,\eta)=0$ implies that $\prec_\zeta $ equals $\prec_\eta $.
\end{lemma}
\begin{proofatend}[\textsc{Proof of \cref{lem_Pseudo} of \cpageref{lem_Pseudo}}]\label{proof_Pseudo}
For any given $a,b\in A$, the set $F_{ab}=\left\{\theta: a\sim_\theta b\right\}$ is closed by \ref{Perfect}.  
Moreover, \ref{Asy} and \ref{Perfect} ensure the existence of a countable and decreasing sequence of open sets with intersection equal to $F_{ab}$.
Since $\Theta$ is normal, the argument of \cref{step_Initial} of \cref{thm_Main} ensures the existence of a continuous function $f_{ab}: \Theta\rightarrow [-1,1]$ such that $f_{ab}^{-1}(0)=F_{ab}$ and $0<f_{ab}(\theta)$ if and only if $a\prec_\theta b$. 

Let $p_{ab}(\zeta,\eta)\defeq \lvert f_{ab}(\zeta)-f_{ab}(\eta)\rvert$ for each $\zeta,\eta\in \Theta$. 
Clearly $p_{ab}: \Theta^2\rightarrow \R$ inherits positivity, symmetry and the triangle inequality from $\lvert\cdot\rvert$ on $\R$. Moreover, $p_{ab}(\zeta,\eta)=0$ implies [$a\prec_\zeta b$ if and only if $a\prec_\eta b$]. (If $b\precsim_\zeta a$ and $a\prec_\eta b$, then $ f_{ab}(\zeta)\leq 0< f_{ab}(\eta)$, so that $p_{ab}(\zeta,\eta)\neq 0$.)

The above argument generates a collection of continuous pseudometrics  $\Pi\defeq \left\{p_{ab}: a,b\in A\right\}$ on $\Theta$. 
Since $A$ is countable, so is $\Pi$. 
Let $\{p_1,p_2,\dots\}$ be any enumeration of $\Pi$ and let  $p\defeq \sum_1^\infty 2^{-n}p_n$. 
Then for every $\zeta, \eta\in \Theta$, $p(\zeta,\eta)=0$ if and only if $p_n(\zeta,\eta)=0$ for every $n$. 

It remains to check that $p$ is indeed a continuous pseudometric. 
Since each $p_n$ is nonnegative and symmetric with values in $[0,2]$, so is $p$. 
Moreover, for each $m$, the partial sum $\sum_1^m2^{-n}p_n$ satisfies the triangle inequality by induction since the sum of two pseudometrics preserves this inequality. 
The sandwich or squeeze lemma for limits of sequences then ensures $p$ also satisfies the triangle inequality.  
Continuity of $p$ follows by uniform convergence of the continuous partial sums. 
This completes the proof of the lemma.
\end{proofatend}

\section{Extensions to joint continuity}
\label{sec_Discussion}\label{sec_Joint_continuity}
In sub\cref{sec_Discussion_axioms} we noted that similar results in the literature assume the closed graph axiom \ref{CG}. 
Moreover, recall that this axiom requires that $A$ is a topological space and that $A\times \Theta$ is endowed with the product topology. 
Thus, the first and most basic difference between the results of \cref{sec_Parametric_continuity} and related results in the literature on joint continuity (\citet{Hildenbrand_Joint_continuity,Mas-Colell_Joint_continuity,Levin_Joint_continuity,CCM_Joint_continuity}) is that  we do not require that  $A$ or $A\times \Theta$  are topological spaces. 

Of course, in the case where $A$ is a topological space, much more can be said. 
Recall from \cref{prop_Metrizable_Theta} that  \ref{PS}, \ref{Perfect} and \ref{CG} are equivalent when the following three conditions hold:  $A$ is a discrete topological space;  $\Theta$ is perfectly normal; and \ref{Order} holds. 
It stands to reason that when these three conditions hold, we can strengthen the conclusions of \cref{thm_Main} and obtain a continuous representation $U:A\times \Theta\rightarrow \R$. 
In the literature, a continuous representation is said to be jointly continuous. 


\begin{corollary}[proof on \cpageref{proof_Joint_perfect}]\label{cor_Joint_perfect}
Let $A$ be countable and discrete and let $\Theta$ be  perfectly normal. 
Preferences have a (jointly) continuous representation  if and only if both \ref{Order} and \ref{PS}  hold. 
\end{corollary}

\begin{proofatend}[\textsc{Proof of \cref{cor_Joint_perfect} from \cpageref{cor_Joint_perfect}}]\label{proof_Joint_perfect}
Necessity of the axioms for a continuous representation $U:A\times \Theta\rightarrow \R$ follows \cref{thm_Main} and the  fact that joint continuity $U$ is stronger than parametric continuity of $U$. 

The fact that the axioms are sufficient for $U$ jointly continuous is as follows. 
Fix $(a,\theta)\in A\times \Theta$ and, for some directed set $D$, consider an arbitrary net $E=\left((a_\nu,\theta_\nu)\right)_{\nu\in D}$ in $A\times \Theta$ with limit $(a,\theta)$. 
We show that $U(a_\nu,\theta_\nu)\rightarrow U(a,\theta)$. 
Recall that $(a,\theta)$ is the limit of $E$ if and only if, for every neighborhood $N$ of $(a,\theta)$, there exists $\mu\in D$ such that for every $\nu\geq \mu$, $(a_\nu,\theta_\nu)\in N$. 
Since $A$ is discrete, $\{a\}$ is open and for some open neighbourhood $N_\theta$ of $\theta$  in $\Theta$, the set $\left\{a\right\}\times N_\theta$ is an (open) neighborhood of $(a,\theta)$ in the product topology on $A\times \Theta$. 
Thus, there exists $\mu$ such that for every $\nu\geq \mu$, $U(a_\nu,\theta_\nu)=U(a,\theta_\nu)$. 
Finally, \cref{thm_Main} ensures that $U(a,\theta_\nu)\rightarrow U(a,\theta)$. 
\end{proofatend}
An example of a countable set $A$ that is not discrete is the set of rational numbers considered as subspace of $\R$. 
\Cref{cor_Joint_perfect} does not hold for such  sets. 
An alternative proof of \cref{cor_Joint_perfect} ought to be  possible via the following conjecture to extend  the main theorem of \citet{Levin_Joint_continuity}. 
This is because every countable topological space is \emph{second countable} (has a countable basis). 
Also, when $A$ is discrete it is also \emph{locally compact}: each point in $A$ has a compact neighbourhood. 
Finally, 
\citeauthor{Levin_Joint_continuity} takes  $\Theta$ to be metrizable. 

\begin{conjecture}\label{thm_Uncountable}
Let $A$ be second countable and locally compact and let $\Theta$ be perfectly normal. 
Preferences have a continuous representation if and only if both \ref{Order} and \ref{CG} hold.
\end{conjecture}
The grounds for \cref{thm_Uncountable} are the following. 
At the bottom of p.717, \citet{Levin_Joint_continuity}  only uses the fact that metrizable spaces are perfectly normal.
Moreover, when $A$ is second countable and $\Theta$ is perfectly normal, then  $A\times \Theta$ is perfectly normal (\citet[p.249]{Tkachuk_Cp_theory}). 
In turn, $\Theta\times A^2$ is also perfectly normal, so that the set $\left\{(\theta,a,b): a\precsim_\theta b\right\}$ is a zero set. 
A complete proof would need to verify that the remaining arguments of \citet{Levin_Joint_continuity} carry over to the setting where $A\times \Theta$ is perfectly normal. 
(The proof that \ref{CG} is necessary for joint continuity is straightforward.)

It may be possible to extend \citeauthor{Levin_Joint_continuity}'s theorem to allow for $A\times \Theta$ that is not perfectly normal without strengthening \ref{CG}. 
However, by the following proposition, the best one can hope for is a \emph{separately continuous} representation ($U(a,\cdot)$ is continuous for each $a$ and $U(\cdot,\theta)$ is continuous for each $\theta$). 
For the jointly continuous case, this negative result is analogous to \cref{prop_Upper_bound}.
\begin{proposition}[proof on \cpageref{proof_Sorgenfrey}]\label{prop_Sorgenfrey}
If $A\times \Theta$ is not perfectly normal, then there are preferences with no continuous representation that satisfy both \ref{Order} and  \ref{CG}. 
\end{proposition}
\begin{proofatend}[\textsc{Proof of \cref{prop_Sorgenfrey} of \cpageref{prop_Sorgenfrey}}]\label{proof_Sorgenfrey}
By assumption, there exists a closed, nonzero subset $F$ of $A\times \Theta$. 
Let $\left\{(a,\theta): a\sim_\theta b\right\}=F$ and let preferences satisfy \ref{Order} and \ref{CG} on $A\bs \{b\}$. 
Then every representation has $U(a,\theta)-U(b,\theta)=0$ for every $(a,\theta)\in F$. 
Let $U'$ be the following transformation of $U$. 
For every $a\in A$, $U'(a,\cdot)=U(a,\cdot)-U(b,\cdot)$. 
Then $U': A\times \Theta\rightarrow \R$ satisfies  $U'(F)=0$. 
That is, $(a,\theta)\in F$ implies $U'(a,\theta)=0$. 
Let $b\prec_\theta a$ for every $(a,\theta)$ in the open set $(A\times \Theta)\bs F$. 
Since $F$ is closed and $b\precsim_\theta a$ for every $(a,\theta)\in A\times \Theta$, preferences satisfy \ref{Order} and \ref{CG} on all of $A$.  
Since $U'(b,\cdot)$ is identically equal to zero, $0<U'(a,\theta)$ for every $(a,\theta)\not\in F$. 
Since $F=(U')^{-1}(0)$ is not a zero set, $U'$ is discontinuous on $A\times \Theta$. 

For an explicit example consider the Sorgenfrey line $\mbb L$. 
This is the unit interval $I$ where the basic open sets are half-open intervals $[r,s)$ such that $r<s$  in $I$. 
$\mbb L$ is a well-known example of a perfectly normal,   separable space that is not second countable and such that the Sorgenfrey plane $\mbb L^2$ is not normal. 
Take $A$ to be the discrete union of $\mbb L$ and $\{b\} $ for some $b\not\in \mbb L$ and take $\Theta=\mbb L$. 
Finally, take $F$ to be the anti-diagonal of $\mbb L^2$  and let $\left\{\prec_\Theta\right\}$ be such that for each $-r\in\Theta$, $r$ is the worst element in $A\bs \{b\}$; $\prec_{-r}$ assigns higher order to elements that are further from $r$ according to the standard metric on $\R$; and, moreover, for each feasible $\epsilon>0$,  $-\epsilon+r\sim_{-r} \epsilon+r$. 
Finally, for every rational number $q\in \mbb L$, let $b\sim_{-q} q$; and for every irrational number $s\in \mbb L$, suppose that $b\prec_{-s} s$. 

Clearly $\{\prec_\Theta\}$ satisfies \ref{Order}. 
To check \ref{CG}, suppose otherwise that $a_\nu\sim_{\theta_\nu} b$ for every $\nu$ and $(a_\nu,\theta_\nu)\rightarrow (a,\eta)$ such that $b\prec_\eta a$. 
Then by construction, each $\theta_\nu$ is a rational number and $a_\nu=-\theta_\nu$. 
Moreover, since $a_\nu\rightarrow a$ and $\theta_\nu\rightarrow \eta$, we have $a=-\eta$. 
Since the anti-diagonal of $\mbb L^2$ is a discrete, there exists a finite number $\mu$ such that $(a_\nu,\theta_\nu)=(a,\eta)$ for every $\nu\geq \mu$, a contradiction of the assumptions regarding the sequence. 
\end{proofatend}


When $\Theta$ is normal, but not perfect, the product $A\times \Theta$ is not perfectly normal. 
Yet, provided $A$ is countable and discrete and  preferences satisfy \ref{Perfect},  \cref{thm_Normal} and the same proof as that of \cref{cor_Joint_perfect} yield
\begin{corollary}\label{cor_Joint_normal}
Let $A$ be countable and discrete and let $\Theta$ be  normal. 
Preferences have a continuous representation  if and only if both \ref{Order} and \ref{Perfect}  hold. 
\end{corollary}
\section{A final remark}
\Cref{prop_Sorgenfrey} implies that an  extension of \cref{cor_Joint_normal} to uncountable sets of actions requires  a  strengthening of \ref{CG} along the lines of  ``preferences have a perfectly closed graph''. 
In our proof, the fact that the set of actions is countable is essential. 
For instance, the uncountable collection of pseudometrics does not combine to form a single pseudometric (see \citet[Theorem 13]{Kelley_Topology} and the surrounding discussion on uniform spaces). 
This extension appears to be an interesting question that is open to  future research. 

%


\begin{appendices}
\section{Proofs}\label{sec_Proofs}
\printproofs
 \end{appendices}
 \section*{References}
 \bibliographystyle{model5-names}
\bibliography{MyBib}
\end{document}